\documentclass[%
 aip,
 amsmath,amssymb,
reprint, onecolumn 
]{revtex4-1}

\usepackage{graphicx}
\usepackage{dcolumn}
\usepackage{bm}

\usepackage[utf8]{inputenc}
\usepackage[T1]{fontenc}
\usepackage{mathptmx}
\usepackage{etoolbox}

\usepackage{graphicx}
\usepackage{epstopdf, epsfig}
\usepackage{tikz}
\usepackage{amsmath}
\usepackage{amssymb}
\usepackage{amsthm}
\usepackage{amsfonts}
\usepackage{enumerate}

\newtheorem{lemma}{Lemma}

\newtheorem{theorem}{Theorem}

\theoremstyle{definition}
\newtheorem{definition}[theorem]{Definition}

\def\A{{\bm A}}
\def\B{{\bm B}}
\def\E{{\bm E}}
\def\n{{\bm n}}
\def\u{{\bm u}}
\def\v{{\bm v}}
\def\w{{\bm w}}
\def\x{{\bm x}}
\def\y{{\bm y}}

\def\t{{\bm t}}
\def\bzero{{\bm 0}}
\def\bvarphi{{\bm \varphi}}
\def\btheta{{\bm \theta}}
\def\bphi{{\bm \phi}}
\def\bpsi{{\bm \psi}}
\def\d{{\rm d}}
\def\hel{{\mathcal{H}}}

\makeatletter
\def\@email#1#2{%
 \endgroup
 \patchcmd{\titleblock@produce}
  {\frontmatter@RRAPformat}
  {\frontmatter@RRAPformat{\produce@RRAP{*#1\href{mailto:#2}{#2}}}\frontmatter@RRAPformat}
  {}{}
}%
\makeatother
\begin{document}

\preprint{AIP/123-QED}

\title[Relative magnetic helicity under turbulent relaxation]{Relative magnetic helicity under turbulent relaxation}
\author{S. Lindberg}
\email{sauli.lindberg@helsinki.fi.}
\affiliation{ 
Department of Mathematics and Statistics, University of Helsinki, P.O. Box 68, 00014 Helsingin yliopisto, Finland
}%
\author{D. MacTaggart}%
 \email{david.mactaggart@glasgow.ac.uk.}
\affiliation{ 
School of Mathematics and Statistics, University of Glasgow, Glasgow, G12 8QQ, UK
}%


\begin{abstract}
Magnetic helicity is a quantity that underpins many theories of magnetic relaxation in electrically conducting fluids, both laminar and turbulent. Although much theoretical effort has been expended on magnetic fields that are everywhere tangent to their domain boundaries, many applications, both in astrophysics and laboratories, actually involve magnetic fields that are line-tied to the boundary, i.e. with a non-trivial normal component on the boundary. This modification of the boundary condition requires a modification of magnetic helicity, whose suitable replacement is called relative magnetic helicity. In this work, we investigate rigorously the behaviour of relative magnetic helicity under turbulent relaxation. In particular, we specify the normal component of the magnetic field on the boundary and consider the \emph{ideal magnetohydrodynamic limit} of resistivity tending to zero in order to model the turbulent evolution in the sense of Onsager's theory of turbulence. We show that relative magnetic helicity is conserved in this distinguished limit and that, for constant viscosity, the magnetic field relaxes in a weak, time-averaged sense to a magnetohydrostatic equilibrium.
\end{abstract}

\maketitle

\section{Introduction}
Magnetic relaxation describes the evolution of magnetic and velocity fields in electrically-conducting fluids, under magnetohydrodynamics (MHD), as they relax to some equilibrium end state (assuming such a state exists). One feature of theories of magnetic relaxation is to avoid the full details of the initial value problem (which is highly non-linear and could involve a series of complicated instabilities and/or turbulence) and predict the end state based on a subset of known properties of the system. The theory of magnetic relaxation, \emph{par excellence}, is \emph{Taylor relaxation} \cite{Tay74}. The key ingredient of this theory is \emph{magnetic helicity}, which is an invariant of ideal MHD. Taylor argued that while a turbulent relaxation would result in magnetic reconnection throughout the domain, the global magnetic helicity would nonetheless be approximately conserved. Using this conservation property in conjunction with Woltjer's theorem \cite{Wol58}, which states that a linear force-free field is the minimizer of magnetic energy for a given value of magnetic helicity, Taylor's theory allows for the prediction of the relaxed end state - a linear force-free field whose force-free parameter is determined by the magnetic helicity (and the magnetic energy). This result is remarkable given that nothing about the turbulent relaxation itself need be known explicitly. All that is assumed about the relaxation is that it is efficient enough to eliminate the kinetic energy of the system, thus allowing for a (quasi-)static equilibrium to be reached.

Extensions to Taylor relaxation have introduced the effects of other constraints upon the system \citep[e.g.][]{BD82,BDGCW83,Ham14}. However, the general spirit of these approaches is similar - the prediction of the final state is based on a variational calculation in which the magnetic energy is minimized subject to given constraints.

Although most theoretical work is based on magnetic fields which are tangent to their domain boundaries, many applications involve magnetic fields with non-zero normal components fixed on the domain boundary. This change in boundary conditions requires a change in the form of magnetic helicity that is needed. We will expand upon the details later, but, for now, it suffices to say that classical magnetic helicity needs to be replaced by \emph{relative magnetic helicity} (hereafter relative helicity) in order to have a gauge-invariant quantity that is conserved under ideal MHD \cite{BF84,FA85,MV23}. It has been shown that Woltjer's theorem can be adapted for line-tied magnetic fields \cite{AL91}, in which the minimizer of the magnetic energy for a fixed value of relative helicity is a linear force-free field.

Magnetic relaxation theories have also been developed for magnetic fields that are not everywhere tangent to the boundary, in a similar vein to the classical case \cite[e.g.][]{ADC02}. These studies are complemented by numerical experiments \cite[e.g.][]{YRG15,Pon16} that approximate the full initial value problem. What is missing among these approaches is a rigorous description of the behaviour of relative helicity and turbulent magnetic relaxation for magnetic fields that are not just tangent to the domain boundaries.

It is in this area that this paper makes a contribution. By focussing on \emph{line-tied} magnetic fields (fixed normal components on the boundary), we provide rigorous results for the behaviour of relative helicity under turbulent relaxation and the nature of relaxed states. We will begin by providing some motivation for our approach to MHD turbulence, which is based on Onsager's theory of ``ideal turbulence'' \cite{Ons49}. We then present the setup of our problem and describe the fundamental turbulent solutions that form the basis of our results. This section is followed by some further verification of our approach, by showing that energy and relative helicity behave as expected in a resistive regime. After this, we present our main results on MHD turbulent relaxtion in the ideal MHD limit. The paper concludes with a summary of our results and a short discussion.

\section{Motivation of  ``ideal turbulence''}
The approach of setting a magnetic relaxation theory in terms of a variational principle, as mentioned in the previous section, allows one to essentially ignore the details of the turbulent relaxation. However, as simulations \cite{YRG15,Pon16,Mof15} and some theoretical works \cite{Kom22,ConP23} indicate, a relaxation may not lead to a linear force-free field and so the evolution must be taken into consideration more directly. Our purpose here is to provide rigorous statements about relative helicity and turbulent relaxation, in order to complement the results of simulations. Therefore, we need a suitable description of MHD turbulence that captures the main properties of turbulence while being amenable to rigorous analysis. The approach we adopt is a basic extension of Onsager's theory of turbulence to MHD. Since this topic has been well-explained in other works \cite[e.g.][]{BV21,Ey24,DLS19}, we only highlight some key features here and direct the reader to other works for more information.

Onsager suggested that \emph{singular} or \emph{weak} solutions of the ideal fluid (Euler) equations would represent a good model of turbulence.  There has been much rigorous work in this area but, for now, let us stick to a qualitative and intuitive description of three key features that will provide some justification of modelling MHD turbulence with weak solutions in the ideal MHD limit.

\subsection{Anomalous dissipation}
An important property of turbulence is the so-called \textit{anomalous dissipation}. This describes the important feature of turbulence in which energy dissipation becomes essentially independent of the diffusive parameters in the system. Let us consider an incompressible electrically-conducting fluid with viscosity $\nu$ and resistivity $\eta$, in a domain $\Omega$. If $\B$ and $\u$ denote the magnetic and velocity fields respectively, then, under anomalous dissipation, we can expect that when $\nu$, $\eta\to 0$,
\begin{equation}
    \nu\int_0^T \int_\Omega |\nabla\u|^2\,\d^3x \, \d t + \eta\int_0^T \int_\Omega |\nabla\B|^2\,\d^3x \to \varepsilon > 0.
\end{equation}
From this statement, we are led naturally to the conclusion that $(\u,\B)$ is not smooth, satisfying
\begin{equation}
    \int_\Omega |\nabla\u|^2\,\d^3x\sim\frac{1}{\nu}\to\infty \quad \text{or} \quad  \int_\Omega |\nabla\B|^2\,\d^3x\sim\frac{1}{\eta}\to\infty.
\end{equation}
Anomalous dissipation is, therefore, a property of the ``roughness'' of $\u$ and $\B$ rather than the dissipative terms themselves. Onsager's original conjecture was that weak solutions of the Euler equations would dissipate energy if the velocities are in a H\"{o}lder space $C^\alpha$ with $\alpha\le1/3$. Systematic study via the method of convex integration was begun by De Lellis and Sz\'ekelyhidi\cite{DLS09,DLS13}, and after several seminal works the conjecture was proved by Isett \cite{Ise18} and by Buckmaster, De Lellis, Sz\'ekelyhidi and Vicol\cite{BDLSV19}. Onsager's conjecture also transfers to MHD \cite{Cal97,BV21}. 

\subsection{Fast magnetic reconnection}
A key feature of Taylor's theory is that there is \emph{efficient} or, as it is referred to in the MHD literature, \emph{fast} reconnection \cite{Priest14}. By ``fast'', we indicate that the process of reconnection does not depend on some slow diffusive timescale and, in the turbulent case, becomes essentially independent of $\eta$. For suitably smooth $\B$ and $\u$, no reconnection is possible in ideal MHD - the field lines are frozen into the flow and there are no changes in connectivity. However, by relaxing the smoothness of the vector fields, flux-freezing can break down \cite{EyAl06}. In this way, reconnection may occur even in the ideal MHD limit \cite{Ey11,Ey13N}. Similar to anomalous dissipation, this characteristic is no longer related to the dissipative terms in the MHD equations, but to the ``roughness'' of the fields.

\subsection{Magnetic helicity conservation}

Another important feature of MHD turbulence is the robust conservation of magnetic helicity (which, as we will describe later, is closely related to efficient reconnection in turbulent relaxation). Helicity conservation has been observed in many simulations (both laminar and turbulent), but it is also a feature of ideal turbulence. All theoretical studies for weak solutions, to date, have focussed on classical magnetic helicity, for which the magnetic field is everywhere tangent to the boundary or exists in a triply periodic box. It has been shown that, for weak solutions of the ideal MHD equations, magnetic helicity is already conserved at the $L^3$ integrability level \cite{KL07,Alu09}. However, magnetic helicity is not conserved, in general, only for finite total energy \cite{BeeBuckVic20,FLS24}. In other words, what is needed is that the magnetic and velocity fields satisfy
\begin{equation}\label{L3_bound}
    \int_0^T\int_\Omega(|\u|^3+|\B|^3)\,\d^3x\,\d t < \infty,
\end{equation}
rather than
\begin{equation}\label{L2_bound}
    \int_0^T\int_\Omega(|\u|^2+|\B|^2)\,\d^3x\,\d t < \infty,
\end{equation}
with the latter being the total energy accumulated over time $T$.

Weak solutions to the ideal MHD equations are highly non-unique \cite{BeeBuckVic20,FLS21}, and there is a danger of constructing solutions that are not physically viable. Faraco and Lindberg\cite{FL20} addressed this by considering solutions with more regularity - weak solutions of ideal MHD that arise in the ideal MHD limit $\eta\to0$ of \emph{Leray-Hopf solutions} (see Definition \ref{d:Leray-Hopf}). With this extra regularity, Faraco and Lindberg were able to prove that (classical) magnetic helicity is conserved in the ideal MHD limit, despite the fact that only the total energy (and not the $L^3$-norm) of the solutions can be controlled. We will adapt and extend their results to relative helicity and line-tied magnetic fields.

\section{Problem setup}
Now that we have provided some motivation for the use of Onsager's approach to turbulence, we now focus on the MHD turbulent relaxation of line-tied magnetic fields. Consider an incompressible fluid contained in a bounded, possibly multiply connected, domain $\Omega \subset \mathbb{R}^3$ with a smooth boundary. The evolution of the fluid is assumed to be governed by the viscous and resistive MHD equations
\begin{eqnarray}
    \frac{\partial\u}{\partial t} + (\u\cdot\nabla)\u &=& -\nabla p + (\nabla\times\B)\times\B + \nu\Delta\u, \label{mhd_1}\\
    \frac{\partial\B}{\partial t} &=& \nabla\times(\u\times\B) + \eta\Delta\B, \label{mhd_2}\\
    \nabla\cdot\u &=& \nabla\cdot\B = 0, \label{mhd_3}\\
    \u(\x,0) &=& \u_0, \quad \B(\x,0) = \B_0, \label{mhd_4}
\end{eqnarray}
for which we have set $\rho=\mu_0=1$ for simplicity. In addition to the variables already defined, $p$ is the fluid pressure and $\u_0$ and $\B_0$ are the initial velocity and magnetic fields.

On the boundary $\Gamma := \partial\Omega$, we  assume that the plasma velocity satisfies the no-slip condition $\u|_\Gamma = \bzero$ and that $\Gamma$ is perfectly conducting, $\E \times \n = (\B \times \u + \eta \nabla \times \B) \times \n = \bzero$. These assumptions imply the further conditions that $(\nabla \times \B) \times \n = \bzero$ and that $\B(\x,t) \cdot \n(\x) = g(\x)$ is stationary, so that (\ref{mhd_1})-(\ref{mhd_4}) are complemented with
\begin{equation} \label{mhd_5}
\u|_{\Gamma} = \bzero, \quad (\B-\B_0) \cdot \n|_\Gamma = 0, \quad (\nabla \times \B) \times \n|_\Gamma = \bzero.
\end{equation}
Equations (\ref{mhd_1})-(\ref{mhd_5}) form the basic model of MHD, with line-tied magnetic fields, that we will study; they have also been studied under the moniker of the \emph{Ohm-Navier-Stokes equations}\cite{YG83,YUI84}. For turbulent solutions, we will consider solutions of these equations in the ideal MHD limit $\eta\to0$. 

\subsection{Relative helicity}
As mentioned earlier, we require relative magnetic helicity for line-tied magnetic fields. Since we wish our results to be applicable in both astrophysical and laboratory environments, we need to consider relative helicity in multiply, as well as simply, connected domains. The relevant definition of relative helicity in multiply connected domains was recently developed by \cite{MV23} and has the form
\begin{equation}\label{hel}
    \hel = \int_\Omega(\A+\A')\cdot(\B-\B')\,\d^3x -\sum_{i=1}^g\left(\oint_{\gamma_i}(\A+\A')\cdot\t\,\d x\right)\left(\int_{\Sigma_i}(\B-\B')\cdot\n_{\Sigma_i}\,\d^2 x\right),
\end{equation}
where $\A$ and $\A'$ are vector potentials of $\B$ and $\B'$ respectively, such that $\B=\nabla\times\A$ and $\B'=\nabla\times\A'$. The quantities related to the geometry of the domain are: $g$, the genus of the domain (typically, the number of holes in the domain for laboratory applications); $\gamma_i$, the $i$th closed path around a hole, on $\Gamma$, with corresponding unit tangent vector $\t_i$; $\Sigma_i$, the $i$th cross-sectional cut in $\Omega$ with corresponding unit normal vector $\n_{\Sigma_i}$. More detailed descriptions of this geometrical setup can be found in \cite{FLMV21} and \cite{MV23}.

If the domain is simply connected ($g=0$) or the flux of $\B$ matches that of the reference field $\B'$ on the cross-sectional cuts, equation (\ref{hel}) reduces to the so-called Finn-Antonsen formula \cite{FA85}, which is used widely in the literature.

In this work, we fix the reference field $\B'$ to be a stationary potential field whose fluxes over the cross-sectional cuts match those of $\B$ (and, therefore, those of $\B_0$), i.e.
\begin{equation} \label{potential_field}
    \nabla \times \B' = \bzero, \quad \nabla \cdot \B' = 0, \quad (\B' - \B) \cdot \n|_\Gamma = 0, \quad \int_{\Sigma_i} (\B'-\B) \cdot \n_{\Sigma_i}\,\d^2 x = 0.
\end{equation}
We make this choice so that, under resistive MHD, $\hel(t)\to0$ as $t\to\infty$, since a potential field is the minimum-energy field for a given distribution of the normal component of the magnetic field on $\Gamma$. 

\subsection{Leray-Hopf solutions}
Our results are based on Leray-Hopf solutions, which are weak solutions of the viscous, resistive MHD equations that obey an energy inequality. Before introducing these solutions, modified from their standard form in order to accommodate line-tied magnetic fields, we present a brief overview of the necessary function spaces and their notation.

\subsubsection{Function spaces}
We use the standard notation of $L^p(\Omega)$ to represent the space of $p$-integrable (vector) functions.  We also make use of $H^1(\Omega)$  to represent the Sobolev space of functions in $L^2(\Omega)$ with one weak derivative also in $L^2(\Omega)$. If $\v \in H^1(\Omega)$ vanishes on the boundary $\Gamma$, we denote $\v \in H^1_0(\Omega)$; see e.g. \cite{Eva10}.

Many of the following results will depend on divergence-free fields that are tangent to $\Gamma$, even if one of the novel aspects of our work is to present results for fields that do not satisfy this condition. For fields tangent to the boundary, we will make particular use of the space
\[
L^2_\sigma(\Omega) = \{\v\in L^2(\Omega);\,\, \nabla\cdot\v=0,\,\,\v\cdot\n|_\Gamma=0\}.
\]
Whenever we use $\sigma$ for a subscript to a function space, this indicates that the elements of the space are divergence-free and tangent to the boundary. If 0 is included in the subscript, this means that each element of the space vanishes on the boundary (which is the case for $\u$). A superscript asterisk attached to a function space denotes its dual. For instance, $\|\w\|_{H^1_\sigma(\Omega)^*} = \max_{\|\v\|_{H^1_\sigma(\Omega)}=1} \langle \w,\v\rangle$.

When the domain $\Omega$ is multiply connected, $L^2_\sigma(\Omega)$ can be expressed, via a Hodge-type decomposition, as 
\[
L^2_\sigma(\Omega) = L^2_\Sigma(\Omega) \oplus L^2_H(\Omega),
\]
where the space of harmonic Neumann vector fields is
\[
L^2_H(\Omega) = \{\v\in L^2_\sigma(\Omega):\,\,\nabla\times\v = \bzero\}, 
\]
and the space of zero-flux vector fields is 
\[
L^2_\Sigma(\Omega) = \left\{\v\in L^2_\sigma(\Omega):\,\,\int_{\Sigma_i}\v\cdot\n_{\Sigma_i}\,\d^2x=0\,\,{\rm for}\,\,i=1,\dots, g \right\}.
\]
We will make use of this decomposition in defining Leray-Hopf solutions for line-tied magnetic fields.

For time-dependent vector fields, given $1\le p<\infty$ and a Banach space $X$, we denote by $L^p(0,T;X)$ and $L^\infty(0,T;X)$ the Bochner spaces of Bochner-integrable functions $\v:\,\,(0,T)\to X$ satisfying
\begin{eqnarray*}
    \int_0^T\|\v(\cdot,t)\|_X^p\,\d t < \infty, \quad \|\|\v(\cdot,t)\|_X\|_{L^\infty(0,T)} < \infty,
\end{eqnarray*}
respectively. If $\v \in L^\infty(0,T;L^2_\sigma)$ has the extra property that $\v(t_j) \rightharpoonup \v(t)$ in $L^2_\sigma$ whenever $t_j \to t \in [0,T)$, we denote $\v \in C_w([0,T);L^2_\sigma(\Omega))$.

\subsubsection{Main definitions}
We have now developed enough notation to be able to write down the definition of line-tied Leray-Hopf solutions for resistive MHD.

\begin{definition} \label{d:Leray-Hopf}
Let $\u_0 \in L^2_\sigma(\Omega)$ and $\B_0 \in L^2(\Omega)$ with $\nabla \cdot \B_0 = 0$. Suppose that $\u \in C_w([0,T);L^2_\sigma(\Omega)) \cap L^2(0,T;H^1_0(\Omega))$ and $\B-\B' \in C_w([0,T);L^2_\Sigma(\Omega)) \cap L^2(0,T; H^1(\Omega))$ satisfy $\partial_t \u \in L^1(0,T;(H^1_{0,\sigma}(\Omega))^*)$ and $\partial_t \B \in L^1(0,T;(H^1_\sigma(\Omega))^*)$, and that

\begin{eqnarray}
    \left\langle \frac{\partial\u}{\partial t}, \bvarphi \right\rangle + \int_{\Omega} (\u \cdot \nabla \u - (\nabla \times \B) \times \B) \cdot \bvarphi\,\d^3x + \nu \int_\Omega \nabla \u : \nabla \bvarphi\,\d^3x &=& 0, \label{lh_weak_1} \\
    \left\langle \frac{\partial\B}{\partial t}, \btheta \right\rangle + \int_{\Omega} \B \times \u \cdot \nabla \times \btheta\,\d^3x + \eta \int_\Omega \nabla \times \B \cdot \nabla \times \btheta\,\d^3x = 0 \label{lh_weak_2}
\end{eqnarray}
hold at a.e. $t \in [0,T)$ and every $\boldsymbol{\varphi} \in H^1_{0,\sigma}(\Omega)$ and $\boldsymbol{\theta} \in H^1_\sigma(\Omega)$. Suppose, furthermore, that $\u(\cdot,0) = \u_0$ and $\B(\cdot,0) = \B_0$ and that $\u$ and $\B$ satisfy the \emph{energy inequality}
\begin{eqnarray}
   & &\displaystyle \frac{1}{2} \int_{\Omega} (|\u(\x,t)|^2 + |\B(\x,t)|^2) \, \d^3x \nonumber\\[4pt]
   &+&\displaystyle \int_s^t \int_{\Omega} (\nu |\nabla  \u(\x,\tau)|^2 + \eta |\nabla \times  \B(\x,\tau)|^2) \, \d^3x \, \d\tau \nonumber\\[4pt]
   &\le& \displaystyle \frac{1}{2} \int_{\Omega} (|\u(\x,s)|^2 + |\B(\x,s)|^2) \, \d^3x, \label{lh_energy}
\end{eqnarray}
for all $t > s$ and almost all $s \geq 0$, including $s = 0$. Then $(\u,\B)$ is called a Leray-Hopf solution of equations (\ref{mhd_1})-(\ref{mhd_5}).
\end{definition}

This is the form of weak solution of the MHD equations that will form the basis of our later results. We also define ideal MHD limits of Leray-Hopf solutions.

\begin{definition}
Suppose $\u_0$ and $\B_0$ are as above, $\nu > 0$ is fixed and $\eta_j \to 0$. For each $j$ denote the corresponding Leray-Hopf solution by $(\u_j,\B_j)$. If $\|\u_j-\u\|_{L^2(0,T;L^2)} \to 0$ and $\|\B_j-\B\|_{L^2(0,T;L^2)} \to 0$ for each $T > 0$, we say that $\u \in C_w([0,\infty);L^2_\sigma)$ and $\B \in C_w([0,\infty);L^2)$ \emph{arise at the ideal MHD limit} $\eta \to 0$.
\end{definition}


\subsection{Existence of Leray-Hopf solutions}
We now prove the existence of Leray-Hopf solutions, as we have just defined them. In the following, we closely follow the approach described in \cite{FL20}, whilst making some adaptations for accommodating line-tied magnetic fields.

\begin{theorem} \label{t:Leray-Hopf}
    For every $\u_0 \in L^2_\sigma(\Omega)$ and $\B_0 \in L^2(\Omega)$ with $\nabla \cdot \B_0 = 0$ there exists a Leray-Hopf solution $(\u,\B)$ of \eqref{mhd_1}--\eqref{mhd_5}. Furthermore,
   \begin{equation}\label{hel_decay}
        \hel(t) = \hel(0) -2\eta\int_0^t\int_\Omega\B\cdot\nabla\times\B\,\d^3x\,\d\tau \quad \text{for all } t \in [0,\infty).
        \end{equation}
\end{theorem}

To set the stage for the proof, following \cite{YG83,YUI84}, we rewrite the MHD equations (\ref{mhd_1})-(\ref{mhd_5}) in terms of the magnetic field $\B_f = \B-\B'$, related to the free magnetic energy, and seek solutions $(\B_f,\u)$ of
\begin{eqnarray}
    &&\frac{\partial\u}{\partial t} + (\u\cdot\nabla)\u = -\nabla p + (\nabla\times\B_f)\times(\B_f+\B') + \nu\Delta\u, \label{mhd_free_1}\\
    &&\frac{\partial\B_f}{\partial t} = \nabla\times[\u\times(\B_f+\B')] + \eta\Delta\B_f, \label{mhd_free_2}\\
    &&\nabla\cdot\u = \nabla\cdot\B_f = 0, \label{mhd_free_3}\\
    &&\u(\x,0) = \u_0, \quad \B_f(\x,0) = \B_0-\B', \label{mhd_free_4} \\
    &&\u|_{\Gamma} = \bzero, \quad \B_f \cdot \n|_\Gamma = 0, \quad (\nabla \times \B_f) \times \n|_\Gamma = \bzero. \label{mhd_free_5}
\end{eqnarray}
The benefit of rewriting the MHD equations in terms of $\B_f$ is that this field is tangential to the boundary and we can make use of the proof for when $\B\cdot\n|_\Gamma=0$ \cite{FL20}. We will solve for $\u \in L^\infty(0,\infty;L^2_\sigma(\Omega)) \cap L^2(0,\infty;H^1_0(\Omega))$ and $\B_f \in L^\infty(0,\infty;L^2_\Sigma(\Omega)) \cap L^2(0,\infty;H^1(\Omega))$.

In order to produce a Galerkin approximation of the MHD equations (\ref{mhd_free_1})-(\ref{mhd_free_5}), we first appeal to the following two results:

\begin{lemma} \label{Stokes lemma}
$L^2_\sigma(\Omega)$ has an orthonormal basis $\{\v_j\}_{j \in \mathbb{N}}$ with the following properties: for every $j \in \mathbb{N}$ there exists $\lambda_j > 0$ such that $\v_j \in H^1_{0,\sigma}(\Omega)$ satisfies
\begin{equation}
(\v_j, \bphi)_{H^1_{0,\sigma}} = \lambda_j (\v_j, \bphi)_{L^2}
\end{equation}
for all $\bphi \in H^1_{0,\sigma}(\Omega)$. Furthermore, $\{\v_j/\sqrt{\lambda_j}\}_{j \in \mathbb{N}}$ is an orthonormal basis of $H^1_{0,\sigma}(\Omega)$.
\end{lemma}

\begin{lemma}\label{Magnetostatic lemma}
$L^2_\Sigma(\Omega)$ has an orthonormal basis $\{\w_j\}_{j \in \mathbb{N}}$ with the following properties: for every $j \in \mathbb{N}$ there exists $\mu_j > 0$ such that $\w_j \in H^1_\sigma(\Omega)$ satisfies
\begin{equation} \label{Magnetostatic equation}
(\w_j, \bpsi)_{H^1_\sigma} = \mu_j (\w_j, \bpsi)_{L^2}
\end{equation}
for all $\bpsi \in H^1_\sigma(\Omega)$. Furthermore, $\{\w_j/\sqrt{\mu_j}\}_{j \in \mathbb{N}}$ is an orthonormal system in $H^1_\sigma(\Omega)$.
\end{lemma}
For the proofs of these two lemmas see \cite{FL20} and the references contained therein.

\begin{proof}[Proof of Theorem \ref{t:Leray-Hopf}]
Projecting onto a subspace of $n$ eigenvectors, with the operators
\[
P_n \u = \sum_{j=1}^n \left( \u, \frac{\v_j}{\sqrt{\lambda_j}} \right)_{H^1_{0,\sigma}} \frac{\v_j}{\sqrt{\lambda_j}}, \quad Q_n \B = \sum_{j=1}^n \left( \B, \frac{\w_j}{\sqrt{\mu_j}} \right)_{H^1_\sigma} \frac{\w_j}{\sqrt{\mu_j}},
\]
the $n$th order Galerkin approximation is 
\begin{eqnarray}
&& \frac{\d}{\d t} (\u_n,\v_j)_{L^2} + \nu (\nabla \u_n,\nabla \v_j)_{L^2} + \langle (\u_n \cdot \nabla) \u_n - (\nabla \times \B_n) \times \B_n, \v_j\rangle_{(H^1_{0,\sigma})^*-H^1_{0,\sigma}} \nonumber\\[4pt]
&&= \langle (\nabla \times \B_n) \times \B',\v_j \rangle_{(H^1_{0,\sigma})^*-H^1_{0,\sigma}}, \\[4pt]
&& \frac{\d}{\d t} (\B_n,\w_j)_{L^2} + \eta (\nabla \times \B_n,\nabla \times \w_j)_{L^2} + \langle \nabla \times (\B_n \times \u_n),\w_j \rangle_{(H^1_\sigma)^*-H^1_\sigma} \nonumber\\[4pt]
&&= - \langle \nabla \times (\B' \times \u_n), \w_j \rangle_{(H^1_\sigma)^*-H^1_\sigma}, \\ [4pt]
&& \u_n(\cdot,0) = P_n \u_0, \quad \B_n(\cdot,0) = Q_n(\B_0-\B'),
\end{eqnarray}
where
\[
\u_n(\x,t) = \sum_{j=1}^n c_{nj}(t) \v_j(\x), \quad \B_n(\x,t) = \sum_{j=1}^n d_{nj}(t) \w_j(\x),
\]
with $c_{nj}$, $d_{nj}\in C^1([0,T))$. 

By denoting
\begin{eqnarray*}
&& \alpha_{jkl} := \int_\Omega (\v_k \cdot \nabla) \v_l \cdot \v_j\,\d^3x, \quad \beta_{jkl} := \int_\Omega (\w_k \cdot \nabla) \w_l \cdot \v_j\,\d^3x, \\
&& \gamma_{jkl} := \int_\Omega (\w_k \cdot \nabla) \v_l \cdot \w_j\,\d^3x, \quad \delta_{jkl} := \int_\Omega (\v_k \cdot \nabla) \w_l \cdot \w_j\,\d^3x, \\
&& \epsilon_{j\ell} := \int_\Omega (\nabla \times \w_\ell) \times \B' \cdot \v_j\,\d^3x,
\end{eqnarray*}
the Galerkin equations obtain the form
\begin{eqnarray}
& \dot{c}_{nj}(t) - \nu \lambda_j c_{nj}(t) + \displaystyle\sum_{k,l=1}^n c_{nk}(t) c_{nl}(t) \alpha_{jkl} - \sum_{k,l=1}^n d_{nk}(t) d_{nl}(t) \beta_{jkl} \nonumber\\
&= \displaystyle\sum_{l=1}^n d_{n\ell}(t) \epsilon_{jl}, \\
& \dot{d}_{nj}(t) - \eta \mu_j d_{nj}(t) + \displaystyle\sum_{k,l=1}^n c_{nk}(t) d_{nl}(t) \gamma_{jkl} - \sum_{k,l=1}^n d_{nk}(t) c_{nl}(t) \delta_{jkl} \nonumber\\
&= -\displaystyle\sum_{l=1}^n c_{n\ell}(t) \epsilon_{jl}, \\
& c_{nj}(0) = (\u_0,\v_j)_{L^2}, \quad d_{nj}(0) = (\B_0-\B', \w_j)_{L^2}. \nonumber
\end{eqnarray}
By the standard theory of ODEs, a unique smooth solution exists for all times, and the energy equality obtains the form
\begin{eqnarray}
&& \sum_{j=1}^n [c_{nj}(t)^2 + d_{nj}(t)^2] + 2 \sum_{j=1}^n \int_s^t [\nu \lambda_j c_{nj}(\tau)^2 + \eta \mu_j d_{nj}(\tau)^2]\, \d \tau \nonumber\\
&&= \sum_{j=1}^n [c_{nj}(s)^2 + d_{nj}(s)^2]
\end{eqnarray}
or, written another way,
\begin{eqnarray}\label{lh_ineq_2}
&&\frac{1}{2} \int_\Omega (|\u_n(t)|^2+|\B_n(t)|^2)\,\d^3x + \int_s^t \int_\Omega (\nu |\nabla \u|^2 + \eta |\nabla \times \B_n|^2)\,\d^3x\,\d\tau \nonumber \\[4pt] 
&&= \frac{1}{2} \int_\Omega (|\u_n(s)|^2+|\B_n(s)|^2)\,\d^3x.  
\end{eqnarray}

By standard arguments, we pass to the weak limit in $L^2(0,T;H^1(\Omega))$ at each $T > 0$ to obtain a Leray-Hopf solution $\u \in L^\infty(0,\infty;L^2_\sigma) \cap L^2(0,\infty;H^1_0)$ and $\B_f \in L^\infty(0,\infty;L^2_\Sigma) \cap L^2(0,\infty;H^1)$ with
\begin{eqnarray}
&&\frac{1}{2} \int_\Omega (|\u(t)|^2+|\B_f(t)|^2)\,\d^3x + \int_s^t \int_\Omega (\nu |\nabla \u|^2 + \eta |\nabla \times \B_f|^2)\,\d^3x\,\d\tau \nonumber\\[4pt]
&&\leq \frac{1}{2} \int_\Omega (|\u(s)|^2 + |\B_f(s)|^2)\,\d^3x, \label{lh_energy_ineq1}
\end{eqnarray}
which can be rewritten as the strong energy inequality
\begin{eqnarray}
&& \frac{1}{2} \int_{\Omega} (|\u(\x,t)|^2 + |\B(\x,t)|^2) \, \d^3x \nonumber\\
&&+ \int_s^t \int_{\Omega} (\nu |\nabla \u(\x,\tau)|^2 + \eta |\nabla \times \B(\x,\tau)|^2) \, \d^3x \, d\tau \nonumber\\
&&\le \frac{1}{2} \int_{\Omega} (|\u(\x,s)|^2 + |\B(\x,s)|^2) \, d^3x,
\end{eqnarray}
for all $t > s$ and a.e. $s \geq 0$, including $s = 0$. As in \cite{FL20}, for a subsequence, $\partial_t \u_n \rightharpoonup \partial_t \u$ in $L^{4/3}(0,T;(H_{0,\sigma}^1)^*)$ and $\partial_t \B_n \rightharpoonup \partial_t \B$ in $L^{4/3}(0,T;(H_\sigma^1)^*)$ for every $T > 0$.

It remains to prove formula \eqref{hel_decay}. Since we are considering $\B-\B' \in L^\infty(0,T;L^2_\Sigma) \cap L^2(0,T;H^1)$, equation (\ref{hel}) reduces to the Finn-Antonsen formula
\begin{equation}
    \hel = \int_\Omega(\A+\A')\cdot(\B-\B')\,\d^3x,
\end{equation}
since the fluxes of $\B$ and $\B'$ match on the cross-sectional cuts $\Sigma_i$ ($i=1,\dots,g$) in $\Omega$. We can then follow \cite{FL20} to get \eqref{hel_decay}.
\end{proof}

\section{Ideal turbulent MHD}
We now move to ideal turbulence, as described earlier, by taking the ideal MHD limit $\eta\to0$. This distinguished limit has been studied previously in the context of \emph{laminar} magnetic relaxation. Moffatt \cite{Mof85} used the system (\ref{mhd_1})-(\ref{mhd_4}) with $\eta = 0$, $\u|_\Gamma = \bzero$ and $\B \cdot \n|_\Gamma = 0$ to investigate the end states of magnetic relaxation. He used the energy equality
\begin{equation} \label{e:Energy equality}
\frac{\d }{\d t} \int_\Omega \frac{|\u|^2+|\B|^2}{2} \d x = - \nu \int_\Omega |\nabla \times \u|^2 \d^3x
\end{equation}
(which is valid for smooth solutions) to deduce that total energy is dissipated as long as $\u \not\equiv 0$ and that, being a decreasing function of time, total energy approaches some non-negative number as $t \to \infty$. Moffatt concluded that $\u(t)$ tends to $\boldsymbol{0}$ while $\B(t)$ tends to a limit state $\B_\infty$ which solves the magnetohydrostatic equation $(\nabla \times \B) \times \B - \nabla p = 0$.

Nu\~nez \cite{Nun07} noted that it is not clear why and in what sense the infinite-time limits $\u_\infty = \bzero$ and $\B_\infty$ are attained. He showed, for smooth solutions and under the assumption that $\B$ is uniformly bounded in space-time, that $\|\u(t)\|_{L^2} \to 0$. He noted, however, the possibility that $\B(t)$ does not tend, even weakly, to a definite limit $\B_\infty \in L^2(\Omega)$.

Here we extend this previous analysis for line-tied magnetic fields, paying close attention to the behaviour of relative helicity in this system. In particular, we show that $\u(t)$ does tend to zero (weakly) without the extra integrability assumption that $\B$ is uniformly bounded. We also show that asymptotically, in a time-averaged sense, the magnetic field becomes magnetostatic as $t \to \infty$.

Below, the Lorentz force needs to be understood in a distributional sense via the formula
\begin{equation} \label{e:Lorentz force}
(\nabla \times \B) \times \B = \nabla \cdot (\B \otimes \B) - \nabla |\B|^2/2
\end{equation}
for solenoidal fields; note that the right-hand side makes sense as soon as $\B$ is square integrable. Under the interpretation \eqref{e:Lorentz force}, all the formulas used in the proof are correct although some of them require an approximation argument for justification. As such arguments are standard, for brevity, we exclude them.

\begin{theorem} \label{t:Existence of weak solutions}
    Let $\nu > 0$, $\u_0 \in L^2_\sigma(\Omega)$ and $\B_0 \in L^2(\Omega)$ with $\nabla \cdot \B_0 = 0$. Suppose
\[\u \in C_w([0,\infty);L^2_\sigma) \cap L^2(0,\infty;H^1_0), \quad \B \in C_w([0,\infty);L^2)\]
arise at the ideal MHD limit $\eta \to 0$ of Leray-Hopf solutions of resistive MHD. Then $\B$ and $\u$ have the following properties:
\begin{enumerate}[(i)]
\item $\u$ and $\B$ satisfy equations (\ref{mhd_1})-(\ref{mhd_4}) with $\eta = 0$. \label{i:Ideal limit}

\item The energy inequality
\begin{eqnarray*}
&&\frac{1}{2} \int_\Omega (|\u(\x,t)|^2+|\B(\x,t)|^2) \d^3 x + \nu \int_s^t \int_\Omega |\nabla \times \u(\x,\tau)|^2 \d^3x \d \tau \\
&&\leq \frac{1}{2} \int_\Omega (|\u(\x,s)|^2+|\B(\x,s)|^2) \d^3x
\end{eqnarray*}
holds for all $t > s$ and almost all $s \geq 0$, including $s = 0$. \label{i:Energy inequality at zero resistivity}

\item $\hel(t) = \hel(0)$ for all $t > 0$. \label{i:Taylor}

\item $\u(t) \rightharpoonup 0$ in $L^2(\Omega)$ when $t \to \infty$. \label{i:Limit velocity}

\item As $t \to \infty$, we have $h^{-1} \int_t^{t+h} [(\nabla \times \B(\tau)) \times \B(\tau) - \nabla p(\tau)] \d\tau \to 0$ in $(H^1_{0,\sigma}(\Omega))^*$ for every $h > 0$. \label{i:Magnetostatic limit}

\end{enumerate}
\end{theorem}

\begin{proof}
In (i), the only non-trivial part is the convergence $(\u_j \cdot \nabla) \u_j \to (\u \cdot \nabla) \u$ and $(\nabla \times \B_j) \times \B_j \to (\nabla \times \B) \times \B$ in the sense of distributions. These are obtained by using \eqref{e:Lorentz force} and the formula $(\u \cdot \nabla) \u = \nabla \cdot (\u \otimes \u)$ in a standard manner.

\vspace{0.5cm}

In (ii), for any $0 \leq s \leq t < \infty$, we can pass to a subsequence to have weak convergence $\nabla \u_j \rightharpoonup \nabla \u$ in $L^2(s,t;L^2(\Omega))$, so that $\int_s^t \int_\Omega |\nabla \times \u(\x,\tau)|^2 \d^3 x \d\tau \leq \liminf_{j \to \infty} \int_s^t \int_\Omega |\nabla \times \u_j(\x,\tau)|^2 \d^3 x \d\tau$. Since $\|\u_j-\u\|_{L^2(0,T;L^2)} \to 0$ and $\|\B_j-\B\|_{L^2(0,T;L^2)} \to 0$ for every $T > 0$, by passing to a further subsequence we get $\|\u_j(t) - \u(t)\|_{L^2(\Omega)} \to 0$ and $\|\B_j(t) - \B(t)\|_{L^2(\Omega)} \to 0$ a.e. $t \geq 0$, which gives the energy inequality for a.e. $s,t \geq 0$. By the weak continuity $\B,\u \in C_w([0,\infty);L^2)$, claim \eqref{i:Energy inequality at zero resistivity} follows.

\vspace{0.5cm}

For the proof of \eqref{i:Taylor} we recall that the Leray-Hopf solutions $(\u_j,\B_j)$ of resistive MHD satisfy $\hel(\B_j,t) = \hel(0) - 2 \eta \int_0^t \int_\Omega \B_j(s) \cdot \nabla \B_j(\tau) \,\d^3 x \d \tau$. We argue as in~\cite{FL20}; by the Cauchy-Schwarz inequality and energy inequality,
\begin{align*}
    \left| 2 \eta \int_0^t \int_\Omega \B_j(s) \cdot \nabla \B_j(\tau) \,\d^3 x\, \d \tau \right|
&\leq 2 \eta \|\B_j\|_{L^\infty(0,t;L^2)} \int_0^t \|\nabla \B_j(\tau)\|_{L^2}\, \d \tau \\
&\leq 2 \eta \|\B_0\|_{L^2} \sqrt{t} \sqrt{\int_0^t \|\nabla \B_j(\tau)\|_{L^2}^2\, \d\tau} \\
&\leq \sqrt{2\eta t} \|\B_0\|_{L^2}^2.
\end{align*}
Now the norm convergence $\|\B_j-\B\|_{L^\infty(0,T;L^2(\Omega)} \to 0$ implies $\hel(t) = \hel(0)$ for a.e. $t > 0$. Since $\B \in C_w([0,\infty);L^2)$, the Aubin-Lions lemma gives $\A \in C([0,\infty);L^2)$ and so $\hel(t) = \hel(0)$ for all $t > 0$.

\vspace{0.5cm}

We next prove \eqref{i:Limit velocity}. By Theorem \ref{t:Existence of weak solutions}, $\int_0^\infty \int_{\Omega} |\nabla \times \u|^2 \d x \d t < \infty$. Since $\u|_\Gamma = 0$, we can use the Poincar\'e inequality to get $\|\u(t_k)\|_{L^2} \to 0$ for any sequence of times $t_k \to \infty$ such that $\|\nabla \times \u(t_k)\|_{L^2} \to 0$.

Suppose now that $t_k \to \infty$; since $\sup_{t > 0} \|\u(t)\|_{L^2(\Omega)} < \infty$, it suffices to fix $\boldsymbol{\varphi} \in C_{c,\sigma}^\infty(\Omega)$ and show that $\int_\Omega \u(\x,t_k) \cdot \boldsymbol{\varphi}(\x) \d x \to 0$. Let $\epsilon > 0$. Choose times $t_k' < t_k$ such that $|t_k-t_k'|< \epsilon$ and, for all large enough $k \in \mathbb{N}$, $\|\u(t_k')\|_{L^2} < \epsilon$. We get
\begin{align*}
\left|\int_\Omega (\u(t_k)-\u(t_k')) \cdot \boldsymbol{\varphi}\, \d x \right|
&= \left| \int_{t_k'}^{t_k} \langle \partial_t \u, \boldsymbol{\varphi} \rangle\, \d t \right| \\
&= \left|\int_{t_k'}^{t_k} \int_\Omega (\u \otimes \u - \B \otimes \B - \nu \nabla \u) \cdot \nabla \boldsymbol{\varphi}\, \d x \d t\right| \\
&\leq \epsilon \left\|\nabla \boldsymbol{\varphi}\right\|_{L^\infty} (\|\B\|_{L^\infty(0,\infty;L^2)}^2 + \|\u\|_{L^\infty(0,\infty;L^2)}^2) \\
&+ \nu \sqrt{\epsilon} \|\nabla \boldsymbol{\varphi}\|_{L^2} \|\nabla \u\|_{L^2(0,\infty;L^2)}.
\end{align*}
Claim \eqref{i:Limit velocity} follows since $\epsilon > 0$ was arbitrary.

\vspace{0.5cm}

We then prove \eqref{i:Magnetostatic limit}. Let $h,t > 0$ and $\boldsymbol{\varphi} \in H^1_{0,\sigma}(\Omega)$. We write
\[\int_t^{t+h} \int_\Omega ((\nabla \times \B) \times \B - \nabla p) \cdot \boldsymbol{\varphi} \,\d^3x \d\tau = \int_t^{t+h} \langle \partial_\tau \u + (\u \cdot \nabla) \u - \nu \Delta \u, \boldsymbol{\varphi} \rangle \d\tau.\]
By using \eqref{i:Limit velocity} we get
\[\frac{1}{h} \int_t^{t+h} \langle \partial_\tau \u, \boldsymbol{\varphi} \rangle\,\d^3x \d\tau = \frac{1}{h} \int_\Omega (\u(t+h)-\u(t)) \cdot \boldsymbol{\varphi} \,\d^3x \to 0.\]
Furthermore, using the formula $(\u \cdot \nabla) \u = \nabla \cdot (\u \otimes \u)$ as well as H\"older and Poincar\'e inequalities we obtain
\begin{align*}
\left|\frac{1}{h} \int_t^{t+h} \int_\Omega (\u \cdot \nabla) \u \cdot \boldsymbol{\varphi} \,\d^3x \d\tau\right|
&= \left|\frac{1}{h} \int_t^{t+h} \int_\Omega u_i u_j \partial_j \boldsymbol{\varphi}_i \,\d^3x \d\tau\right| \\
&\leq \frac{1}{h} \|\u\|_{L^2(t,t+h;L^4(\Omega))}^2 \|\boldsymbol{\varphi}\|_{H^1(\Omega)} \\
&\leq \frac{C}{h} \|\nabla \u\|_{L^2(t,t+h;L^2(\Omega))}^2 \|\boldsymbol{\varphi}\|_{H^1(\Omega)}\to 0.
\end{align*}
Similarly,
\[
\left|\frac{1}{h} \int_t^{t+h} \langle \Delta \u, \boldsymbol{\varphi} \rangle \,\d^3x \d\tau\right|
= \left|\frac{1}{h} \int_t^{t+h} \int_\Omega \nabla \u \cdot \nabla \boldsymbol{\varphi} \,\d^3x \d\tau\right| \leq \frac{1}{h^{1/2}} \|\nabla \u\|_{L^2(t,t+h;L^2(\Omega))} \|\boldsymbol{\varphi}\|_{H^1(\Omega)} \to 0.
\]
The proof of Theorem \ref{t:Existence of weak solutions} is complete.\end{proof}



\section{Summary}

The purpose of this work has been to provide some rigorous results about the MHD relaxation of line-tied magnetic fields, with a particular focus on relative helicity. Our results extend those of \cite{Mof85} and \cite{Nun07}, not only in considering a different, and more applicable, boundary condition (with all the consequences for magnetic helicity of such a change) but also in translating the results into the context of turbulence. In particular, we have shown, through Theorem \ref{t:Existence of weak solutions}, that in the ideal MHD limit, relative helicity is conserved under turbulent relaxation and that the end state tends asymptotically, in a time-averaged sense, to a (weak) magnetohydrostatic balance.

\section*{Acknowledgments}
SL acknowledges support by the ERC Advanced Grant 834728 and the Academy of Finland CoE FiRST. DM acknowledges
support from a Leverhulme Trust grant (RPG-2023-182), a  Science and Technologies Facilities Council (STFC) grant (ST/Y001672/1) and a Personal Fellowship from the Royal Society of Edinburgh (ID: 4282). 

\bibliography{relhel}

\end{document}